\documentclass[a4paper,USenglish,cleveref,autoref,thm-restate]{lipics-v2021}
\usepackage{color}
\usepackage{todonotes}
\usepackage{complexity}
\usepackage{xspace}
\usepackage{thm-restate}
\usepackage{paralist}
\title{On the Parameterized Complexity of Computing \\$st$-Orientations with Few Transitive Edges\thanks{A preliminary version of this paper will appear in the proceedings of MFCS 2023. Research partially supported by (i) University of Perugia, Fondi di Ricerca di Ateneo, edizione 2021, project ``AIDMIX - Artificial Intelligence for Decision making: Methods for Interpretability and eXplainability'', (ii) Progetto RICBA22CB {\em ``Modelli, algoritmi e sistemi per la visualizzazione e l'analisi di grafi e reti''.}}}
\titlerunning{On the Parameterized Complexity of Computing $st$-Orientations}

\newcommand{\tw}{\omega}
\newcommand{\sto}{\textsc{st-Orientation}\xspace}
\newcommand{\ntstolong}{\textsc{Non-Transitive st-Orientation}}
\newcommand{\ntsto}{\textsc{NT-st-Orientation}\xspace}
\newcommand{\naelong}{\textsc{Not-all-equal 3-Sat}\xspace}
\newcommand{\nae}{\textsc{NAE-3-Sat}\xspace}
\newcommand{\true}{\text{true}}
\newcommand{\false}{\text{false}}

\hideLIPIcs
\nolinenumbers

 \author{Carla Binucci}{Department of Engineering, University of Perugia, Italy}{carla.binucci@unipg.it}{https://orcid.org/0000-0002-5320-9110}{}

 \author{Giuseppe Liotta}{Department of Engineering, University of Perugia, Italy}{giuseppe.liotta@unipg.it}{https://orcid.org/0000-0002-2886-9694}{}

 \author{Fabrizio Montecchiani}{Department of Engineering, University of Perugia, Italy}{fabrizio.montecchiani@unipg.it}{https://orcid.org/0000-0002-0543-8912}{}

 \author{Giacomo Ortali}{Department of Engineering, University of Perugia, Italy}{giacomo.ortali@unipg.it}{https://orcid.org/0000-0002-4481-698X}{}

 \author{Tommaso Piselli}{Department of Engineering, University of Perugia, Italy}{tommaso.piselli@studenti.unipg.it}{https://orcid.org/0000-0002-7088-920X}{}

 \authorrunning{C. Binucci, G. Liotta, F. Montecchiani, G. Ortali, T. Piselli} 

 \Copyright{Carla Binucci, Giuseppe Liotta, Fabrizio Montecchiani, Giacomo Ortali, Tommaso Piselli} 



\begin{document}

\maketitle

\begin{abstract}
    Orienting the edges of an undirected graph such that the resulting digraph satisfies some given constraints is a classical problem in graph theory, with multiple algorithmic applications. In particular, an $st$-orientation orients each edge of the input graph such that the resulting digraph is acyclic, and it contains a single source $s$ and a single sink $t$. Computing an $st$-orientation of a graph can be done efficiently, and it finds notable applications in graph algorithms and in particular in graph drawing. On the other hand, finding an $st$-orientation with at most $k$ transitive edges is more challenging and it was recently proven to be \NP-hard already when $k=0$. We strengthen this result by showing that the problem remains \NP-hard even for graphs of bounded diameter, and for graphs of bounded vertex degree. These computational lower bounds naturally raise the question about which structural parameters can lead to tractable parameterizations of the problem. Our main result is a fixed-parameter tractable algorithm parameterized by treewidth.
\end{abstract}

\section{Introduction}
An orientation of an undirected graph is an assignment of a direction to each edge, turning the initial graph into a directed graph (or \emph{digraph} for short). Notable examples of orientations are acyclic orientations, which guarantee the resulting digraph to be acyclic; transitive orientations, which make the resulting digraph its own transitive closure; and Eulerian orientations, in which each vertex has equal in-degree and out-degree. 
Of particular interest for our research are certain constrained acyclic orientations, which find applications in several domains, including graph planarity and graph drawing. More specifically, given a graph $G=(V,E)$ and two vertices $s,t \in V$, an \emph{$st$-orientation} of $G$, also known as \emph{bipolar orientation}, is an orientation of its edges such that the corresponding digraph is acyclic and contains a single source $s$ and a single sink $t$. It is well-known that $G$ admits an $st$-orientation if and only if it is biconnected after the addition of the edge $st$ (if not already present). The computation of an $st$-numbering (an equivalent concept defined on the vertices of the graph) is for instance part of the quadratic-time planarity testing algorithm by Lempel, Even and Cederbaum~\cite{lempel1967algorithm}. Later, Even and Tarjan~\cite{EVEN1976339} showed how to compute an $st$-numbering in linear time, and used this result to derive a linear-time planarity testing algorithm. In the field of graph drawing, bipolar orientations are a central algorithmic tool to compute different types of layouts, including visibility representations, polyline drawings, dominance drawings, and orthogonal drawings (see~\cite{DBLP:books/ph/BattistaETT99,DBLP:conf/dagstuhl/1999dg} for references). On a similar note, a notable result states that a planar digraph admits an upward planar drawing if and only if it is the subgraph of a planar $st$-graph, that is, a planar digraph with a bipolar orientation~\cite{DBLP:journals/tcs/BattistaT88}.

\begin{figure}[t]
    \centering
    \begin{minipage}[t]{0.25\textwidth}
    \includegraphics[width=\columnwidth,page=3]{st-orientations.pdf}
    \subcaption{}
    \end{minipage}\hfill
    \begin{minipage}[t]{0.25\textwidth}\includegraphics[width=\columnwidth,page=1]{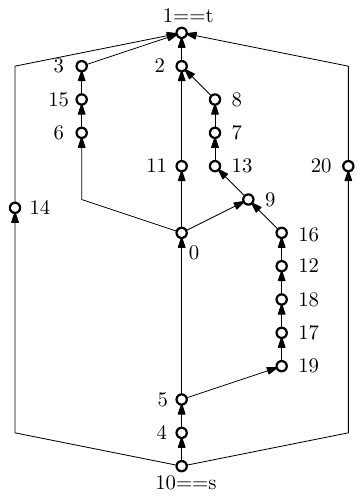}
    \subcaption{}
     \end{minipage}\hfill
     \begin{minipage}[t]{0.25\textwidth}\includegraphics[width=\columnwidth,page=2]{st-orientations.pdf}\subcaption{\label{fig:introc}}
     \end{minipage}
    \caption{(a): An undirected graph $G$ with randomly labeled vertices. (b)-(c): Two polyline drawings of $G$ computed by using different $st$-orientations. The drawing in (b) uses an $st$-orientation without transitive edges and it has smaller area and number of bends than the drawing in (c). 
}
    \label{fig:intro}
\end{figure}

Recently, Binucci, Didimo and Patrignani~\cite{DBLP:conf/gd/BinucciDP22} focused on $st$-orientations with no transitive edges. We recall that an edge $uv$ is \emph{transitive} if the digraph contains a path directed from $u$ to $v$; for example, the bold (red) edges in \cref{fig:introc} are transitive, see also \cref{se:preliminaries} for formal definitions. Besides being of theoretical interest, such orientations, when they exist, can be used to compute readable and compact drawings of  graphs~\cite{DBLP:conf/gd/BinucciDP22}. For example, a classical graph drawing algorithm relies on $st$-orientations to compute polyline representations of  planar graphs. The algorithm is such that both the height and the number of bends of the representations can be reduced by computing $st$-orientations with few transitive edges. See Algorithm {\tt Polyline} in~\cite{DBLP:books/ph/BattistaETT99} for details and \cref{fig:intro} for an example. 
 
Unfortunately,  while an $st$-orientation of an $n$-vertex graph can be computed in $O(n)$ time, computing one that has the minimum number of transitive edges is much more challenging from a computational perspective.  Namely, Binucci et al.~\cite{DBLP:conf/gd/BinucciDP22} prove that the problem of deciding whether an $st$-orientation with no transitive edges exists is \NP-complete, and provide an ILP model for planar graphs. 

\begin{figure}[t]
    \centering
    \includegraphics[scale=0.9]{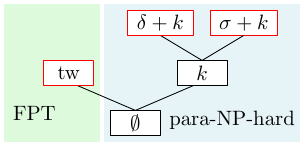}
    \caption{The complexity landscape of the \sto problem. The symbols tw, $\delta$, and $\sigma$ denote the treewidth, the maximum vertex degree, and the diameter of the graph, respectively. The boxes with red boundaries denote the new results presented in this paper. }
    \label{fig:st}
\end{figure}

\smallskip\noindent\textbf{Contribution.} We study the parameterized complexity of finding $st$-orientations with few transitive edges. More formally, given a graph $G$ and an integer $k$, the \sto problem asks for an $st$-orientation of $G$ with at most $k$ transitive edges (see also \cref{se:preliminaries}). As already discussed, \sto is para-\NP-hard by the natural parameter $k$~\cite{DBLP:conf/gd/BinucciDP22}. We strengthen this result by showing that, for $k=0$, \sto remains \NP-hard even for graphs of diameter at most six, and for graphs of vertex degree at most four. In light of these computational lower bounds, we seek for structural parameters that can lead to tractable parameterizations of the problem. Our main result is a fixed-parameter tractable algorithm for \sto parameterized by treewidth, a central parameter in the parameterized complexity analysis (see~\cite{DBLP:series/mcs/DowneyF99,DBLP:journals/jal/RobertsonS86}). \cref{fig:st} depicts a summary of the computational complexity results known for the \sto problem.

It is worth remarking that by Courcelle's theorem one can derive an (implicit) \FPT~algorithm parameterized by treewidth and $k$, while we provide an explicit algorithm parameterized by treewidth only. The main challenge in applying dynamic programming over a tree decomposition is that the algorithm must know if adding an edge to the graph may cause previously forgotten non-transitive edges to become transitive, and, if so, how many of them. To tackle this difficulty, we describe an approach that avoids storing information about all edges that may potentially become transitive; instead, the algorithm guesses the edges that will be transitive in a candidate solution and ensures that no other edge will become transitive in the course of the algorithm. Our technique can be easily adapted to handle more general constraints on the sought orientation, for instance the presence of multiple sources and sinks.

\smallskip\noindent\textbf{Paper structure.} We begin with preliminary definitions and basic tools, which can be found in \cref{se:preliminaries}. In \cref{se:algo} we describe our main result, an \FPT~algorithm for the \sto problem parameterized by treewidth. \cref{se:hardness} contains our second contribution, namely we adapt the \NP-hardness proof in~\cite{DBLP:conf/gd/BinucciDP22} to prove that the result holds also for graphs that have bounded diameter and for graphs with bounded vertex degree. In the latter case, the graphs used in the reduction not only have bounded vertex degree (at most four), but are also subdivisions of triconnected graphs. In \cref{se:conclusions} we list some interesting open problems that stem from our research.

\section{Preliminaries}\label{se:preliminaries}
\noindent \textbf{Edge orientations.} Let $G=(V,E)$ be an undirected graph. An \emph{orientation} $O$ of $G$ is an assignment of a \emph{direction}, also called \emph{orientation}, to each edge of $G$. We denote by $D_O(G)$ the digraph obtained from $G$ by applying the orientation $O$. For each undirected pair $(u,v) \in E$, we write $uv$ if $(u,v)$ is oriented from $u$ to $v$ in $D_O(G)$, and we write $vu$ otherwise. A directed path from a vertex $u$ to a vertex $v$ is denoted by $u \leadsto v$. A vertex of $D_O(G)$ is a \emph{source} (\emph{sink}) if all its edges are outgoing (incoming). An edge $uv$ of $D_O(G)$ is \emph{transitive} if $D_O(G)$ contains a directed path $u \leadsto v$ distinct from the edge $uv$. A digraph $D_O(G)$ is an \emph{$st$-graph} if: (i) it contains a single source $s$ and a single sink $t$, and (ii) it is acyclic. An orientation $O$ such that $D_O(G)$ is an $st$-graph is called an \emph{$st$-orientation}. 


\medskip\noindent\fbox{%
  \parbox{0.98\textwidth}{
\sto\\
\textnormal{\textbf{Input:}} An undirected graph $G=(V,E)$, two vertices $s,t \in V$, and an integer $k \ge 1$.\\
\textnormal{\textbf{Output:}} An $st$-orientation $O$ of $G$ such that the resulting digraph $D_O(G)$ contains at most $k$ transitive edges.
  }%
}

\smallskip\noindent We recall that \sto is \NP-complete already for $k=0$~\cite{DBLP:conf/gd/BinucciDP22}, which hinders tractability in the parameter $k$. Also, in what follows, we always assume that the input graph $G$ is connected, otherwise we can immediately reject the instance as any orientation would give rise to at least one source and one sink for each connected component of $G$.

\smallskip\noindent \textbf{Tree-decompositions.}  Let $(\mathcal{X},T)$ be a pair such that $\mathcal{X}=\{X_i\}_{i \in [\ell]}$ is a collection of subsets of vertices of a graph $G=(V,E)$, called \emph{bags}, and $T$ is a tree whose nodes are in  one-to-one correspondence with the elements of $\mathcal X$. When this creates no ambiguity, $X_i$ will denote both a bag of $\mathcal{X}$ and the node of $T$ whose corresponding bag is $X_i$. The pair $(\mathcal{X},T)$ is a \emph{tree-decomposition} of $G$ if:
\begin{inparaenum}
\item $\bigcup_{i\in [\ell]} X_i = V$,
\item For every edge $uv$ of $G$, there exists a bag $X_i$ that contains both $u$ and~$v$, and
\item For every vertex $v$ of $G$, the set of nodes of $T$ whose bags contain $v$ induces a non-empty (connected) subtree of $T$.
\end{inparaenum}
The \emph{width} of  $(\mathcal{X},T)$ is $\max_{i=1}^\ell {|X_i| - 1}$, while the \emph{treewidth} of $G$, denoted by $\textnormal{tw}(G)$, is the minimum width over all tree-decompositions of $G$. The problem of computing a tree-decomposition of width $\textnormal{tw}(G)$ is fixed-parameter tractable in $\textnormal{tw}(G)$~\cite{DBLP:journals/siamcomp/Bodlaender96}. 
 A tree-decomposition $(\mathcal{X},T)$ of a graph $G$ is \emph{nice} if $T$ is a rooted binary tree with the following additional  properties~\cite{DBLP:journals/jal/BodlaenderK96}:
 \begin{inparaenum}
 \item If a node  $X_i$ of $T$ has two children whose bags are $X_j$ and $X_{j'}$, then $X_i=X_j=X_{j'}$. In this case, $X_i$ is a \emph{join bag}.
 \item If a node $X_i$ of $T$ has only one child~$X_j$, then $X_i \neq X_j$ and there exists a vertex $v \in G$ such that either $X_i = X_j \cup \{v\}$ or $X_i \cup \{v\} = X_j$. In the former case  $X_i$ is an \emph{introduce bag}, while in the latter case  $X_i$ is a \emph{forget bag}.
 \item If a node $X_i$ is the root or a leaf of $T$, then $X_i=\emptyset$. In this case, $X_i$ is a \emph{leaf bag}.
 \end{inparaenum} 
Given a tree-decomposition of width $\tw$ of $G$, a nice tree-decomposition of $G$ with the same width can be computed in $O(\tw \cdot n)$ time~\cite{DBLP:books/sp/Kloks94}.

\section{The \sto Problem Parameterized by Treewidth}\label{se:algo}
In this section, we describe a fixed-parameter tractable algorithm for \sto parameterized by treewidth. In fact, the algorithm we propose can solve a slightly more general problem. Namely, it does not assume that $s$ and $t$ are part of the input, but it looks for an $st$-orientation in which the source and the sink can be any pair of vertices of the input graph. However, if $s$ and $t$ are prescribed, a simple check can be added to the algorithm (we will highlight the crucial point in which the check is needed) to ensure this property. 

Let $G=(V,E)$ be an undirected graph.  A \emph{solution} of the \sto problem is an orientation $O$ of $G$ such that $D_O(G)$ is an $st$-graph with at most $k$ transitive edges.   Let $(\mathcal{X},T)$ be a tree-decomposition of $G$ of width $\tw$. For a bag $X_i \in \mathcal{X}$, we denote by $G[X_i]$ the subgraph of $G$ induced by the vertices of $X_i$, and by $T_i$ the subtree of $T$ rooted at $X_i$. Also, we denote by $G_i$ the subgraph of $G$ induced by all the vertices in the bags of $T_i$. We adopt a dynamic-programming approach performing a bottom-up traversal of $T$. The solution space is encoded into records associated with the bags of  $T$, which we describe in the next section. 

\subsection{Encoding solutions}
\label{subse:recorddefinition}
Before describing the records stored for each bag, we highlight the main challenges about how to encode the partial solutions computed throughout the course of the algorithm. 
Let $v$ be a vertex introduced in a bag $X_i$. Adding $v$ and its incident edges to a partial solution may either turn many (possibly linearly many) forgotten edges into transitive edges and/or it may make the same forgotten edge transitive with respect to arbitrarily many different paths. This is schematically illustrated in \cref{fi:introduce}, where  $X_i$ and its child bag $X_j$ are highlighted by shaded regions. In \cref{fi:introduce-a}, $e_1, \dots, e_s$ are forgotten edges, i.e., edges in $G_i$ but not in $G[X_i]$; if we orient edge $(u,v)$ from $v$ to $u$ and edge $(v,w)$ from $w$ to $v$ all edges $e_1, \dots, e_s$  become transitive. In \cref{fi:introduce-b}, $e$ is a forgotten edge, while $u_1, \dots, u_s$ and $w_1, \dots, w_h$ are vertices of bag $X_j$; orienting the edges $(w_p,v)$ from $w_p$ to $v$ ($1 \leq p \leq h$) and the edges $(v,u_q)$ from $v$ to $u_q$, turns $e$ into a transitive edge with respect to $h \times s$ different paths. In case of \cref{fi:introduce-a} the algorithm cannot afford reconsidering the forgotten edges as they can be arbitrarily many. In case of \cref{fi:introduce-b} the algorithm should avoid counting $e$ multiple times (for each newly created path). To overcome these issues, the algorithm guesses the edges that are transitive in a candidate solution and verifies that no other edge can become transitive during the bottom-up visit of $T$. This is done by suitable records, describe below.

\begin{figure}[t]
	\centering
	\subfloat[]{\includegraphics[width=0.32\columnwidth,page=9]{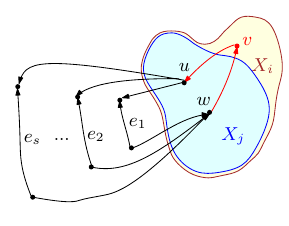}\label{fi:introduce-a}}
	\hfil
	\subfloat[]{\includegraphics[width=0.32\columnwidth,page=10]{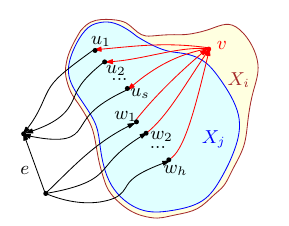}\label{fi:introduce-b}}
	\caption{(a) The directed edges $wv$ and $vu$ make all edges $e_1,...,e_s$  transitive. (b) Each pair of directed edges $w_pv$ and $v_qu$, for $p\in [1,h]$ and $q\in [1,s]$, makes $e$ transitive. }\label{fi:introduce}
\end{figure} 

\smallskip
Let $O$ be a solution and consider a bag $X_i \in \mathcal{X}$. The \emph{record} $R_i$ of $X_i$ that \emph{encodes} $O$ represents the interface of the solution $O$ with respect to $X_i$.  For ease of notation, the restriction of $D_O(G)$ to $G_i$ is denoted by $D_i$, and similarly the restriction to $G[X_i]$ is $D[X_i]$. Record $R_i$ stores the following information.

\begin{itemize}
    \item $\mathcal{O}_i$ which is the orientation of $D[X_i]$.  
    \item $\mathcal{A}_i$ which is the subset of the edges of $D[X_i]$ that are transitive in $D_O(G)$. We call such edges \emph{admissible transitive edges} or simply \emph{admissible edges}. The edges of $G_i$ not in $\mathcal{A}_i$ are called \emph{non-admissible}. We remark that an edge of $\mathcal{A}_i$ may not be transitive in $D_i$. 
    \item $\mathcal{P}_i$ which is the set of ordered pairs of vertices $(a,b)$ such that: (i) $a,b \in X_i$, and (ii) $D_i$ contains the path $a \leadsto b$. 
    \item $\mathcal{F}_i$ which is the set of ordered pairs of vertices $(a,b)$ such that: (i) $a,b \in X_i$, and (ii) connecting $a$ to $b$ with a directed path makes a non-admissible edge of $D_i$ to become transitive.  
    \item $c_i$ which is the \emph{cost} of $R_i$, that is, the number of transitive edges in $D_i$. Note that $c_i \ge |\mathcal{A}_i|$. 
    \item $\mathcal{S}_i$ which maps each vertex $v\in X_i$ to a Boolean value $\mathcal{S}_i(v)$ that is true if and only if $v$ is a source in $D_i$. Analogously, $\mathcal{T}_i$ maps each vertex $v\in X_i$ to a Boolean value $\mathcal{T}_i(v)$ that is true if and only if $v$ is a sink in $D_i$.
    \item $\sigma_i$ which is a flag that indicates whether $D_O(G)$ contains a source that belongs to $G_i$ but not to $X_i$. Analogously, $\tau_i$ is a flag that indicates whether $D_O(G)$ contains a sink that belongs to $G_i$ but not to $X_i$.
\end{itemize}

Observe that, for a bag $X_i$, different solutions $O$ and $O'$ of $G$ may be encoded by the same record $R_i$. In this case, $O$ and $O'$ are \emph{equivalent}. Clearly, this defines an equivalent relation on the set of solutions for $G$, and each record represents an equivalence class. The goal of the algorithm is to incrementally construct the set of records (i.e., the quotient set) for each bag rather than the whole set of solutions. More formally, for each bag $X_i \in \mathcal{X}$, we associate a set of records $\mathcal{R}_i=\{R_i^1,...,R_i^h\}$. While this is not essential for establishing fixed-parameter tractability, we further observe that if more records are equal except for their costs, it suffices to keep in $\mathcal{R}_i$ the one whose cost is no larger than any other record. The next lemma easily follows. 

\begin{lemma}\label{le:records}
For a bag $X_i$, the cardinality of $\mathcal{R}_i$ is $2^{O({\tw^2})}$. Also, each record of $\mathcal{R}_i$ has size $O(\tw^2)$.
\end{lemma}
\begin{proof}
Recall that $G[X_i]$ contains at most $\tw$ vertices and $\tw^2$ edges. 
Observe that the number of possible orientations of the edges of $G[X_i]$ is $O(2^{\tw^2})$. Similarly, the number of possible pairs of vertices (and hence of subsets of edges) of $X_i$ is $O(2^{\tw^2})$. The possible mappings to Boolean values of the vertices in $X_i$ are $O(2^{\tw})$. Hence, a set of distinct records in which there are no two of them that differ only by their cost has size  $2^{O({\tw^2})}$. Finally, the fact that each record has size $O(\tw^2)$ follows directly from the definition.
\end{proof}

\subsection{Description of the algorithm}
\label{subse:descr}

We are now ready to describe our dynamic-programming algorithm over a nice tree-decomposition $(\mathcal{X},T)$ of the input graph $G$. Let $X_i$ be the current bag visited by the algorithm. We compute the records of $X_i$ based on the records computed for its child or children (if any). If the set of records of a bag is empty, the algorithm halts and returns a negative answer. We distinguish four cases based on the type of the bag $X_i$. Observe that, to index the records within $\mathcal{R}_i$, we added a superscript $q \in [h]$ to each record, and we will do the same for all the record's elements.

\smallskip\noindent\textbf{$X_i$ is a leaf bag.} We have that $X_i$ is the empty set and  $\mathcal{R}_i$ consists of only one record, i.e., $\mathcal{R}_i=\{R_i^1=\langle \emptyset, \emptyset, \emptyset, \emptyset, 0, \emptyset,\emptyset, \false, \false \rangle\}$.

\smallskip\noindent\textbf{$X_i$ is an introduce bag.} Let  $X_j =X_i  \setminus \{v\}$ be the child of $X_i$.  Initially, $\mathcal{R}_i=\emptyset$. Next, the algorithm exhaustively extends each record $R_j^p \in \mathcal{R}_j$ to a set of records of $\mathcal{R}_i$ as follows.   Let $\mathcal{O}_v$  be the set of all the possible orientations of the edges incident to $v$ in $G[X_i]$, and similarly let $\mathcal{A}_v$ be the set of all the possible subsets of the edges incident to $v$ in $G[X_i]$. The algorithm considers all possible pairs $(o,t)$ such that $o \in \mathcal{O}_v$ and $t \in \mathcal{A}_v$. For each pair $(o,t)$, we proceed as follows.  
\begin{enumerate}
    \item Let $q=|\mathcal{R}_i|+1$, the algorithm computes a candidate orientation $\mathcal{O}_i^q$ of $G[X_i]$ starting from $\mathcal{O}_j^p$ and orienting the edges of $v$ according to $o$.
    \item  Similarly, it computes the candidate set of admissible edges $\mathcal{A}_i^q$ starting from $\mathcal{A}_j^p$ and adding to it the edges in $t$. 
    \item Next, it sets the candidate cost $c_i^q=c_j^p+|t|$.
    \item Let the \emph{extension} $\langle \mathcal{O}_i^q,\mathcal{A}_i^q,c_i^q \rangle$ be \emph{valid} if: 
    \begin{inparaenum}
        \item $c_i^q\le k$;
        \item there is no pair $(a,b)\in \mathcal{P}_j^p$ so that $bv,va\in D[X_i^q]$;
        \item there is no pair $(a,b)\in \mathcal{F}_j^p$ so that $av,vb\in D[X_i^q]$.
    \end{inparaenum}  
    Clearly, if an extension is not valid, the corresponding record cannot encode any solution; namely, condition (a) ensures that the candidate cost does not exceed $k$, condition (b) guarantees the absence of cycles, condition (c) guarantees that no non-admissible edge becomes transitive. Hence, if an extension is not valid, the algorithm discards it and continues with the next pair $(o,t)$. 
    \item Instead, if the extension is valid, the algorithm computes the record $R_i^q=\langle \mathcal{O}_i^q, \mathcal{A}_i^q, \mathcal{P}_i^q, \mathcal{F}_i^q,$ $c_i^q, \mathcal{S}_i^q, \mathcal{T}_i^q,$ $\sigma_i^q,$ $ \tau_i^q\rangle$ of $\mathcal{R}_i$, where $\sigma_i^q=\sigma_j^p$, $\tau_i^q=\tau_j^p$ (recall that $X_j \subset X_i$). To complete the record $R_i^q$, it remains to compute $\mathcal{S}_i^q$, $\mathcal{T}_i^q$, $\mathcal{P}_i^q$ and $\mathcal{F}_i^q$. 
    \begin{enumerate}
        \item For each vertex $w \in X_j$, we set $\mathcal{S}_i^q(w)=\true$  if and only if $\mathcal{S}_j^p(w)=\true$ and there is no edge of $v$ oriented from $v$ to $w$ in $D[X_i^q]$ (which would make $w$ not a source anymore). Similarly, for each vertex $w \in X_j$, we set $\mathcal{T}_i^q(w)=\true$ if and only if $\mathcal{T}_j^p(w)=\true$ and there is no edge of $v$ oriented from $w$ to $v$ in $D[X_i^q]$. Finally, we set  $\mathcal{S}_i^q(v)=\true$  if and only if $v$ is a source in $D[X_i^q]$ (as by the definition), and we set  $\mathcal{T}_i^q(v)=\true$ if and only if $v$ is a sink in $D[X_i^q]$.
        \item  We initially set $\mathcal{P}_i^q=\emptyset$. We recompute the paths from scratch as follows. We build an auxiliary digraph $D^*$ which we initialize with $D[X_i^q]$. 
        We then add to $D^*$ the information about paths in $\mathcal{P}_j^p$. Namely, for each $(a,b)\in \mathcal{P}_j^p$, we add an edge $ab$ to $D^*$ (if it does not already exists). Once this is done, for each pair $u,w\in X_i \times X_i$ for which there is a path $u \leadsto w$ in $D^*$, we add the pair $(u,w)$ to $\mathcal{P}_i^q$.
        \item Consider now $\mathcal{F}_i^q$. We initially set $\mathcal{F}_i^q=\mathcal{F}_j^p$. Observe that the addition of $v$ might have created new pairs of vertices that should belong to $\mathcal{F}_i^q$. Namely, for each pair $(a,b) \in \mathcal{F}_j^p$, we verify what are the vertices $c$ such that $D[X_i^q]$ contains a path $a \leadsto c$ while $D[X_j^p]$ does not (observe that $a \leadsto c$ contains $v$, possibly $c=v$); for each such vertex, we add $(c,b)$ to $\mathcal{F}_i^q$. See \cref{fi:introduce-5c-a} for an illustration. Similarly, we verify what are the vertices $d$ such that $D[X_i^q]$ contains a path $d \leadsto b$ while $D[X_j^p]$ does not (again $d \leadsto b$ contains $v$, possibly $d=v$); for each such vertex, we add $(a,d)$ to $\mathcal{F}_i^q$. Finally, we consider all the edges incident to $v$ and that are not in $\mathcal{A}_i^q$. These edges are not admissible and we should further update $\mathcal{F}_i^q$ accordingly. This can be done as follows: we consider each edge incident to $v$ not in $\mathcal{A}_i^q$, for each  such an edge $e$ we verify what are the pairs of vertices in $X_i$ (including $e$'s endpoints) such that connecting them with a path makes $e$ transitive, we add such pairs to $\mathcal{F}_i^q$ if not already present. See \cref{fi:introduce-5c-b} for an illustration.
    \end{enumerate} 
\end{enumerate}

\begin{figure}[t]
	\centering
	\subfloat[]{\includegraphics[width=0.4\columnwidth,page=1]{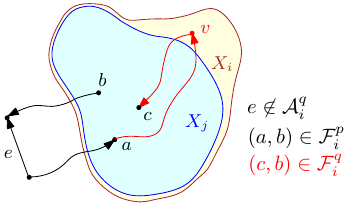}\label{fi:introduce-5c-a}}
	\hfil
	\subfloat[]{\includegraphics[width=0.4\columnwidth,page=2]{introduce-5c.pdf}\label{fi:introduce-5c-b}}
	\caption{Illustration of Step 5c of the algorithm when $X_i$ is an introduce bag.}\label{fi:introduce-5c}
\end{figure}

\smallskip\noindent\textbf{$X_i$ is a forget bag.} Let  $X_j =X_i \cup \{v\}$ be the child of $X_i$. The algorithm computes $\mathcal{R}_i$ by exhaustively merging records of $\mathcal{R}_j$ as follow. 
\begin{enumerate}
    \item For each $R_j^p\in \mathcal{R}_j$, we remove from $\mathcal{O}_j^p$ and $\mathcal{A}_j^p$ all the edges incident to $v$ and from $\mathcal{P}_j^p$ and $\mathcal{F}_j^p$ all the pairs where one of the vertices is $v$. Observe that due to this operation, there might now be records that are identical except possibly for their costs. Among them, we only keep one record whose cost is no larger than any other record.
    \item Let $R_j^p$ be a record of $\mathcal{R}_j$ that was not discarded by the procedure above. If $\mathcal{S}_j^p(v) \wedge \sigma_j^p$, we discard $R_j^p$ (because the encoded orientation would contain two sources), else we set $\sigma_j^p=\true$ (because $v$ is a source). Similarly, if  $\mathcal{T}_j^p(v) \wedge \tau_j^p$, we discard $R_j^p$, else we set $\tau_j^p=\true$. At this point, if the record has not been discarded yet and vertices $s$ and $t$ are prescribed, we can add the following check. If $\mathcal{S}_j^p(v) \wedge \sigma_j^p$, then $v$ is a source, hence if $v \neq s$, we discard the record. Analogously, if $\mathcal{T}_j^p(v) \wedge \tau_j^p$, then $v$ is a sink, hence if $v \neq t$, we discard the record. 
    \item Finally, we remove from $\mathcal{S}_j^p$ and $\mathcal{T}_j^p$ the values $\mathcal{S}_j^p(v)$ and $\mathcal{T}_j^p(v)$.
    \item All the records that have not been discarded and have been updated according to the above procedure are added to $\mathcal{R}_i$. 
\end{enumerate}

\smallskip\noindent\textbf{$X_i$ is a join bag.} Let  $X_j=X_{j'}$ be the two children of $X_i$. The algorithm computes $\mathcal{R}_i$ by exhaustively checking if a pair of records, one from $X_j$ and one from $X_{j'}$, can be merged together. For each pair  $R_j^p$ and $R_{j'}^{p'}$, we proceed as follows.
\begin{enumerate}
    \item We initially set $\mathcal{R}_i=\emptyset$. 
    The two records $R_j^p$ and $R_{j'}^{p'}$ are \emph{mergeable} if: 
    \begin{inparaenum}
        \item $\mathcal{O}_j^p= \mathcal{O}_{j'}^{p'}$;
        \item $\mathcal{A}_j^p= \mathcal{A}_{j'}^{p'}$;
        \item $c_j^p+c_{j'}^{p'}-|\mathcal{A}_j^p| \le k$;
        \item there is no pair $(a,b) \in \mathcal{P}_j^p$ such that $(b,a) \in \mathcal{P}_{j'}^{p'}$;
        \item there is no pair $(a,b) \in \mathcal{P}_j^p$ such that $(a,b) \in \mathcal{F}_{j'}^{p'}$;
        \item there is no pair $(a,b) \in \mathcal{P}_{j'}^{p'}$ such that $(a,b) \in \mathcal{F}_j^p$;
        \item $\neg(\sigma_j^p \wedge \sigma_{j'}^{p'})$;
        \item $\neg(\tau_j^p \wedge \tau_{j'}^{p'})$.
    \end{inparaenum} 
    Conditions \textbf{a}-\textbf{b} are obviously necessary to merge the records. Condition \textbf{c} guarantees that the number of transitive edges (avoiding double counting the admissible edges in $X_i$) is at most $k$. Condition \textbf{d} guarantees the absence of cycles. Conditions \textbf{e}-\textbf{f} guarantee that no non-admissible edge becomes transitive. Conditions \textbf{g}-\textbf{h} guarantee that the resulting orientation contains at most one source and one sink. 
    If the two records are not mergeable, we discard the pair and proceed with the next one. Otherwise we create a new record $R_i^q$, with $q=|\mathcal{R}_i|+1$, and continue to the next step.
    \item Based on the previous discussion, we can now compute $R_i^q$ as follows:
    \begin{inparaenum}
        \item $\mathcal{O}_i^q=\mathcal{O}_j^p$;
        \item $\mathcal{A}_i^q=\mathcal{A}_j^p$;
        \item $c_i^q=c_j^p+c_{j'}^{p'}-|\mathcal{A}_j^p|$;
        \item For each pair $(a,b)$ of vertices of $X_i$, we add it to $\mathcal{P}_i^q$ if it is contained in $\mathcal{P}_j^p$ or in $\mathcal{P}_{j'}^{p'}$.
        \item For each pair $(a,b)$ of vertices of $X_i$, we add it to $\mathcal{F}_i^q$ if $\mathcal{F}_j^p(a,b) \vee \mathcal{F}_{j'}^{p'}(a,b)$.
        \item For each vertex $v$ of $X_i$, we set $\mathcal{S}_i^q(v)= \mathcal{S}_j^p(v) \wedge \mathcal{S}_{j'}^{p'}(v)$;
        \item For each vertex $v$ of $X_i$, we set $\mathcal{T}_i^q(v)= \mathcal{T}_j^p(v) \wedge \mathcal{T}_{j'}^{p'}(v)$;
        \item $\sigma_i^q = \sigma_j^p \vee \sigma_{j'}^{p'}$;
        \item $\tau_i^q = \tau_j^p \vee \tau_{j'}^{p'}$.
    \end{inparaenum}
\end{enumerate}

The next lemma establishes the correctness of the algorithm.

\begin{lemma}\label{le:correctness}
Graph $G$ admits a solution for \sto if and only if the algorithm terminates after visiting the root of $T$. Also, the algorithm outputs a solution, if any.
\end{lemma}
\begin{proof}
($\rightarrow$) Suppose that the algorithm terminates after visiting the root bag $X_\rho$ of $T$. We reconstruct a solution $O$ of $G$ as follows. We can assume that our algorithm  stores additional pointers for each record (a common practice in dynamic programming), such that each record has a single outgoing pointer (and potentially many incoming pointers). Consider a record $R_i^q$ of a bag $X_i$. If $X_i$ is an introduce bag, there is only one record $R_j^p$ of the child bag $X_j$ from which $R_i^q$ was generated and the pointer links $R_i^q$ and $R_j^p$. If $X_i$ is forget bag, there might be multiple records that have been merged into $R_i^q$ and in this case the pointer link $R_i^q$ with one of these records with minimum cost. If $X_i$ is a join bag, there are two mergeable records $R_j^p$ and $R_{j'}^{p'}$ that have been merged together, and the pointer links $R_i^q$ to $R_j^p$ and $R_{j'}^{p'}$. With these pointers at hand, we can apply a top-down traversal of $T$, starting from the single (empty) record of the root bag $X_\rho$ and reconstruct the corresponding orientation $O$. Namely, when visiting an introduce bag and the corresponding record, we orient the edges of the introduced vertex $v$ according to the orientation $O_v$ defined by the record.

We now claim that $D_O(G)$ is an $st$-graph with at most $k$ transitive edges. Suppose first, for a contradiction, that $D_O(G)$ contains more than one source. Let $s$ and $s'$ be two sources of $D_O(G)$. Then $\mathcal{S}_i^q(s)=\false$ in the bag $X_i$ in which $s$ has been forgotten, and similarly for $\mathcal{S}_i^q(s')$. This is however not possible by construction of  $\mathcal{S}_i^q$. Thus,  either the record $R_i^q$ has been discarded because $S_j^p(v) \wedge \sigma_j^p$ (see item \textbf{2} when $X_i$ is a forget bag) or $\sigma_j^p=\false$. The first case contradicts the fact that $R_i^q$ is a record used to reconstruct $O$. The second case implies that $s'$ has not been encountered; however, in this latter case the algorithm sets $\sigma_j^p=\true$, hence some descendant record will be discarded as soon as $s'$ is forgotten, again contradicting the fact that we are considering records with pointers up to the root bag. A symmetric argument shows that $D_O(G)$ contains a single sink. We next argue that $D_O(G)$ is acyclic. Suppose, again for a contradiction, that $D_O(G)$ contains a cycle. In particular, the cycle was created either in an introduce bag or in a join bag. In the former case, let $v$ be the last vertex of this cycle that has been introduced in a bag $X_i$. Let $a,b$ be the neighbors of $v$ that are part of the cycle, and w.l.o.g. assume that the edges are $va$ and $bv$. It must be $\mathcal{P}_i^q$ does not contain the  pair $(a,b)$, otherwise we would have discarded this particular orientation for the edges incident to $v$ (see item \textbf{4.b} when $X_i$ is an introduce bag). On the other hand, one easily verifies that when introducing a vertex $v$, all the new paths involving $v$ are computed from scratch (see item \textbf{5.b} when $X_i$ is an introduce bag), and, similarly, when joining two bags, the existence of a path in one of the two bags is correctly reported in the new record (see item \textbf{2.d} when $X_i$ is a join bag). If the cycle was created in a join bag the argument is analogous, in particular, observe that we verify that there is no path contained in the record of one of the child bags such that the same path with reversed direction exists in the record of the other child bag (see item \textbf{1.d} when $X_i$ is a join bag).  We conclude this direction of the proof by showing that $D_O(G)$ contains at most $k$ transitive edges. Observe first that the cost of the record ensures that at most $k$ edges of $G$ are part of some set of admissible edges. Suppose, for a contradiction, that $D_O(G)$ contains more than $k$ transitive edges. Then there is a bag $X_i$ and a record $R_i^q$ in which a non-admissible edge became transitive. Also, $X_i$ is either an introduce or a join bag. If $X_i$ introduced a vertex $v$, observe that all the newly introduced edges are incident to $v$. On the other hand, the algorithm discarded the orientations of the edges of $v$ for which there is a pair $(a,b) \in \mathcal{F}_j^p$ (with $X_j$ being the child of $X_i$) so that $av,vb \in D[X_i^q]$ (see item \textbf{4.c} when $X_i$ is an introduce bag). Then either the orientation was discarded, which contradicts the fact that we are considering a record used to build the solution, or $\mathcal{F}_j^p$ missed the pair $(a,b)$. Again one verifies this second case is not possible, because the new pairs that are formed in an introduce bag are correctly identified (see item \textbf{5.c} when $X_i$ is an introduce bag) by the algorithm and similarly for join bags (see item \textbf{2.e} when $X_i$ is a join bag). If $X_i$ is a join bag, the argument is analogous, in particular, we verified that there is no path in one of the two child records that makes transitive a non-admissible edge in the other child record (see items \textbf{1.e} and \textbf{1.f} when $X_i$ is a join bag). This concludes the first part of the proof.

\smallskip \noindent
($\leftarrow$) It remains to prove that, if $G$ admits a solution $O$, then the algorithm terminates after visiting the root $X_\rho$ of $T$. If this were not the case, there would be a bag $X_i$ of $T$ and a candidate record  that encodes $O$, such that the record has been incorrectly discarded by the algorithm; we show that this is not possible. Suppose first that $X_i$ is an introduce bag. Then a candidate record is discarded if the cost exceeds $k$, or if a cycle is created, or if a non-admissible edge becomes transitive (see the conditions of item \textbf{4} when $X_i$ is an introduce bag). In all cases the candidate record does not encode a solution. If $X_i$ is a forget bag, we may discard a candidate record if it is identical to another but has a non-smaller cost (see item \textbf{1} when $X_i$ is a forget bag). Hence we always keep a record that either encodes the solution at hand or a solution with fewer transitive edges but with exactly the same interface at $X_i$. Also, we may discard a record if the forgotten vertex $v$ is a source and $G_i$ already contains a source (see item \textbf{2} when $X_i$ is a forget bag). This is correct, because no further edge can be added to $v$ after it is forgotten. A symmetric argument holds for the case in which a record is discarded due to $v$ being a sink. Finally, if $X_i$ is a join bag, pairs of  records of its children bags are discarded if not mergeable (see the conditions of item \textbf{1} when $X_i$ is a join bag). One easily verifies that failing one of the conditions for mergeability implies that the record does not encode a solution (see also the discussion after item \textbf{1}).
\end{proof}

The next theorem summarizes our contribution.

\begin{theorem}label{th:main}
Given an input graph $G=(V,E)$ of treewidth $\omega$ and an integer $k \ge 0$, there is an algorithm that either finds a solution of \sto or reject the input in time $2^{O(\tw^2)}\cdot n$. 
\end{theorem}
\begin{proof}
The correctness of the algorithm has been proved in \cref{le:correctness}. Concerning the time complexity, we begin by using a recent result by Korhonen~\cite{Korhonen21}, which provides a single-exponential  algorithm for computing a $2$-approximation of the treewidth $\textnormal{tw}(G)=\tw-1$ of $G$. Given a tree-decomposition of width $O(\tw)$ of $G$, a nice tree-decomposition of $G$ with the same width can be computed in $O(\tw \cdot n)$ time~\cite{DBLP:books/sp/Kloks94}. 

We now analyze the time complexity of our algorithm for each type of bag. 
Leaf bags are trivially processed in $O(1)$ time. 
For an introduce bag, we iterate over the $2^{O(\tw^2)}$ records of the child bag (see \cref{le:records}), and for each of them we consider all possible extensions, which are again $2^{O(\tw^2)}$. For each valid extension, creating a single record from it takes $\tw^{O(1)}$ time. 
For a forget bag, we update the $2^{O(\tw^2)}$ records of its child bag, and then we iteratively look for pairs of records that can be merged. This takes again $2^{O(\tw^2)}$ time. Also, updating each merged record takes $\tw^{O(1)}$ time. 
For join bags we iteratively look for pairs of records (one for each of the two child bags) that are mergiable. Since there are $2^{O(\tw^2)}$ pairs and checking the mergiability takes $\tw^{O(1)}$ time (as well as eventually computing the merged record), the procedure takes again $2^{O(\tw^2)}$ time. Since $T$ contains $O(n)$ nodes, the statement follows.
\end{proof}

\section{The Complexity of the \ntstolong Problem for Graphs of Bounded Diameter and Bounded Degree}\label{se:hardness}

We begin by recalling the special case of \sto considered in~\cite{DBLP:conf/gd/BinucciDP22}. An $st$-orientation $O$ of a graph $G$ is \emph{non-transitive} if $D_O(G)$ does not contain transitive edges.

\medskip\noindent\fbox{%
  \parbox{0.98\textwidth}{
\ntstolong~(\ntsto)\\
\textnormal{\textbf{Input:}} An undirected graph $G=(V,E)$, and two vertices $s,t\in V$.\\
\textnormal{\textbf{Output:}} An non-transitive $st$-orientation $O$ of $G$ such that vertices $s$ and $t$ are the source and sink of $D_O(G)$, respectively.
  }%
}

\medskip\noindent The hardness proof of \ntsto in~\cite{DBLP:conf/gd/BinucciDP22} exploits a reduction from \naelong (\nae)~\cite{10.1145/800133.804350}. Recall that the input of \nae is a pair $\langle X,\varphi \rangle$ where $X$ is a set of boolean variables and $\varphi$ is a set of clauses, each composed of three literals out of $X$, and the problem asks for an assignment of the variables in $X$ so that each clause in $\varphi$ is composed of at least one true variable and one false variable.

In this section, we show that \ntsto is \NP-hard even for graphs of bounded diameter and for graphs of bounded vertex degree that are subdivisions of triconnected graphs. To prove our results, we first summarize the construction used in~\cite{DBLP:conf/gd/BinucciDP22}. 

\subsection{A Glimpse into the Hardness Proof of  \ntsto}\label{sse:glimpse}

The construction in~\cite{DBLP:conf/gd/BinucciDP22} adopts three types of gadgets, which we recall below. Given an edge $e$ of a digraph $D$ such that $e$ has an end-vertex $v$ of degree 1, we say that $e$ \emph{enters} $D$ if it is outgoing with respect to $v$, and we say that $e$ \emph{exits} $D$ otherwise. Similarly, given a directed edge $e=uv$, we say that $e$ \emph{exits} $u$ and that $e$ \emph{enters} $v$.

\begin{itemize}
\item The \emph{fork gadget} $F$ is depicted in \cref{fi:nphardness-a}. See Lemma~1 of~\cite{DBLP:conf/gd/BinucciDP22}. Namely, if $F$ does not contain $s$ or $t$ (the source and sink prescribed in the input), then in any non-transitive orientation $O$ of a graph $G$ containing $F$, either $e_1$ enters $F$ and $e_9$, $e_{10}$ exit $F$, or vice versa.  \cref{fi:nphardness-a} depicts $F$, $D_{O_1}(F)$ and $D_{O_2}(F)$, where $O_1$ and $O_2$ are the two $st$-orientations admitted by $F$.
\item The \emph{variable gadget} $G_x$ associated to a variable $x\in X$ is shown in \cref{fi:nphardness-b}; observe it contains the designated vertices $s$ and $t$. Its crucial property is stated in Lemma~2 of~\cite{DBLP:conf/gd/BinucciDP22}. Namely, in any non-transitive $st$-orientation $O$ of a graph $G$ containing $G_x$, either $x$ exists $G_x$ and $\overline{x}$ enters $G_x$, or vice-versa. 
\item The \emph{split gadget} $S_k$ is shown in \cref{fi:nphardness-c}; it consists of $k-1$ fork gadgets chained together, for some fixed $k > 0$. The crucial property of this gadget is described in Lemma~3 of~\cite{DBLP:conf/gd/BinucciDP22}. Namely, in any non-transitive $st$-orientation $O$ of a graph $G$ containing $S_k$, either $x$ (the \emph{input edge} of $S_k$) enters $S_k$ and the edges $e_{9}$ and $e_{10}$ of the fork gadgets $F_1,...,F_{k-1}$ incident to one degree-1 vertex (the \emph{outgoing edges} of $S_k$) exit $S_k$, or vice-versa.  
\end{itemize}

\begin{figure}[t]
	\centering
	\subfloat[$F$]{\includegraphics[width=0.33\columnwidth,page=1]{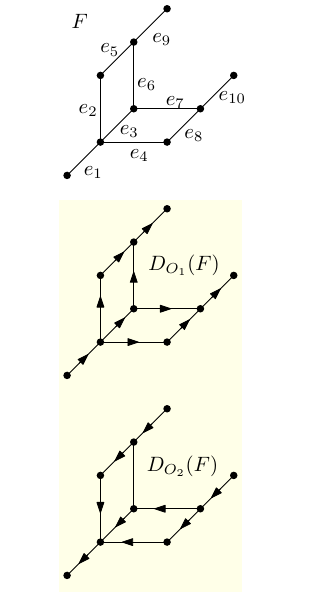}\label{fi:nphardness-a}}
	\hfil
	\subfloat[$G_x$]{\includegraphics[width=0.33\columnwidth,page=2]{st-orientation-NPhardness}\label{fi:nphardness-b}}
	\hfil
	\subfloat[$S_k$]{\includegraphics[width=0.33\columnwidth,page=3]{st-orientation-NPhardness}\label{fi:nphardness-c}}
	\caption{(a) The fork gadget $F$ and its two possible non-transitive $st$-orientations. (b) The variable gadget $D_O(G_x)$ associated to $x\in X$, where $O$ is one of its two possible orientations. (c) The split gadget $D_O(S_k)$ associated to $x$, where $O$ is one of its two possible orientations.}\label{fi:nphardness}
\end{figure}

Given an instance $\langle X,\varphi\rangle $ of \nae, the instance  $\langle G_\varphi, s,t \rangle$ of \ntsto is constructed as follow. For each $x\in X$ we add $G_{x}$ and two split gadgets $S_{k}$ and $S_{\overline{k}}$, where $k$ (resp. $\overline{k}$) is the number of clauses where $x$ appears in its non-negated (resp. negated) form (edges $x$ and $\overline{x}$ are the input edges of $S_k$ and $S_{\overline{k}}$, respectively). Finally, for each clause $c=(x_1,x_2,x_3) \in \varphi$, we add a vertex $c$ that is incident to an output edge of the split gadget of each of its variables. See \cref{fi:nphardness-club-b}, where the non-dashed edges and all the vertices with the exception of $g$ define $G_{\varphi}$. It can be shown that $\langle X,\varphi\rangle $ is a yes-instance of \nae if and only if $\langle G_\varphi, s,t \rangle$ is a yes-instance of \ntsto~\cite{DBLP:conf/gd/BinucciDP22}. 


\subsection{Hardness for Graphs of Bounded Diameter}

Given an undirected graph $G$, the \emph{distance} between two vertices of $G$ is the length of any shortest  path connecting them. The \emph{diameter} of $G$ is the maximum distance over all pairs of vertices of the graph. 
We now adapt the construction in \cref{sse:glimpse} to show that \ntsto remains \NP-hard also for graphs of bounded diameter. We define the \emph{extended fork gadget}  by adding an edge $e_{11}$ to the fork gadget (see \cref{fi:nphardness-club-a}).

\begin{figure}[t]
	\centering
	\subfloat[]{\includegraphics[width=0.24\columnwidth,page=4]{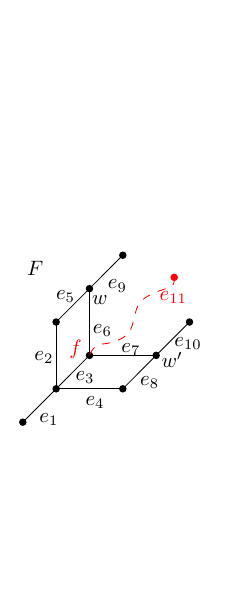}\label{fi:nphardness-club-a}}
	\hfil
	\subfloat[]{\includegraphics[width=0.74\columnwidth,page=5]{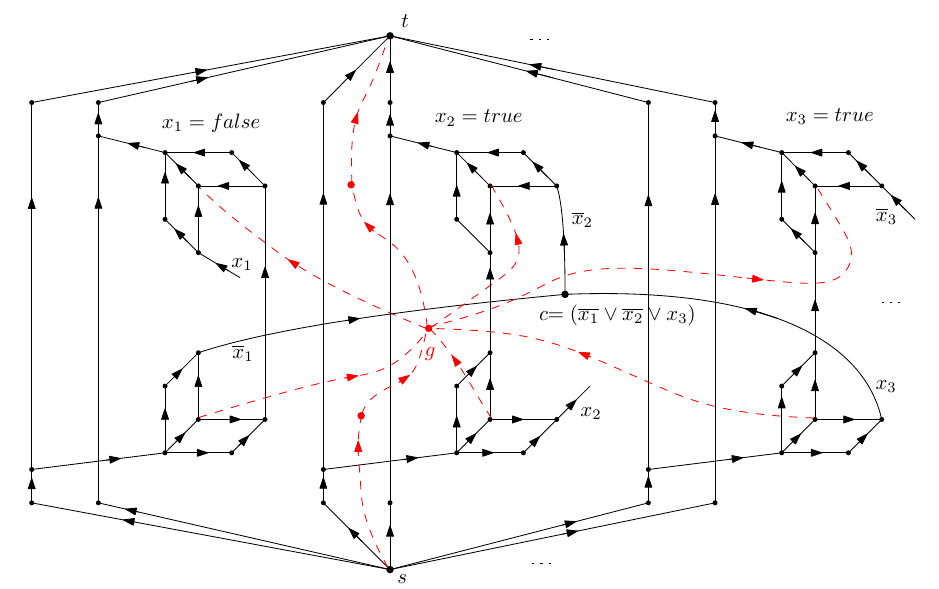}\label{fi:nphardness-club-b}}
	\caption{(a) A fork gadget $F$ extended with edge $e_{11}$. (b) Graphs $D_{O'}(G_{\varphi})$, defined by the non-dashed edges, and graph $D_O(H_{\varphi})$, obtained from $G$ by adding $g$ and the dashed edges. $O'$ and $O$ are non-transitive $st$-orientations of $G_\varphi$ and $H_\varphi$, respectively.  $O$ is obtained by extending~$O'$.
 }\label{fi:nphardness-club}
\end{figure} 

\smallskip \noindent \textbf{Construction of $H_\varphi$.} Given an instance $\langle X,\varphi\rangle $ of \nae and the instance  $\langle G_\varphi, s,t \rangle$ of \ntsto computed as described in \cref{sse:glimpse}, we define $\langle H_\varphi, s,t \rangle$ as follows. We first set $H_\varphi=G_\varphi$. Then, we add a vertex $g$ to $H_\varphi$ and an edge $(g,f)$ for each vertex $f$ belonging to a fork $F$ of $H_\varphi$ and incident to the corresponding edges $e_3$, $e_6$, and $e_7$. Also, we add edges  $(g,t)$ and $(s,g)$, and we subdivide each of them once. See \cref{fi:nphardness-club-b} (the non-dashed edges and all the vertices with the exception of $g$ define~$G_{\varphi}$).

We prove the following technical lemma, in which we use the same notation depicted in \cref{fi:nphardness-club-a}.

\begin{lemma}\label{le:extendedfork}
Let $G$ be an undirected graph containing an extended fork gadget $F$ that does not contain $s$ and $t$. In any non-transitive $st$-orientation $O$ of $G$, either $e_3$ enters $f$ and $e_6,e_7$ exit $f$ or vice versa.
\end{lemma}
\begin{proof}
Suppose, by contradiction, that $e_3$ and $e_6$ enter (resp. exit) $f$. Since $F$ does not contain $s$ and $t$, either the directed path $\langle e_2,e_5 \rangle$ exits, or $\langle e_5,e_2 \rangle$ exists. In the former case,  $e_3$ (resp. $e_6$) is a transitive edge. In the latter case,  $e_6$ (resp. $e_3$) is a transitive edge. In both cases we contradict the fact that the orientation is non-transitive. Hence, either $e_3$ enters $f$ and $e_6$ exits $f$ or vice versa. By using a symmetric argument, the same property can be proved with respect for $e_3$ and $e_7$.
\end{proof}

\begin{theorem}\label{th:degsix}
\ntsto  is \NP-hard for graphs of diameter at most $6$.
\end{theorem}
\begin{proof}
We construct $H_\varphi$ as described above. Observe that any vertex of $G$ is at distance at most 3 to $g$, hence $H_\varphi$ has diameter at most $6$. We show that a non-transitive $st$-orientation of $G_\varphi$ corresponds to a non-transitive $st$-orientation of $H_\varphi$ ($\rightarrow$) and vice versa ($\leftarrow$). 

\smallskip \noindent
($\rightarrow$)  Given a non-transitive $st$-orientation $O'$ of $G_\varphi$, we construct an $st$-orientation $O$ of $H_\varphi$ by extending $O'$ as follow. We orient the four edges of $H_\varphi\setminus G_\varphi$ connecting $s$ to $t$ so that such path is directed from $s$ to $g$. For each other edge $e$, which is incident to $g$, we orient it so that $e$ enters $g$ if and only if $e$ is the edge incident to an extended fork gadget whose corresponding edge $e_1$ is an entering edge. See \cref{fi:nphardness-club-b}. For each two vertices $a,b\in D_O(H_\varphi)$, there is no path $a \leadsto b$ so that $ag,gb\in D_O(H_\varphi)$. Hence, since $D_{O'}(G_\varphi)$ has no cycle, also $D_O(H_\varphi)$ has no cycle. Consequently, $O$ is an acyclic orientation with $s$ and $t$ being its single source and sink, respectively. We now show that it does not contain transitive edges. Let $e=ab$ be any edge of $D_O(H_\varphi)$. We have that any  path from $s \leadsto t$ containing $g$ either contains edges incident to degre-2 vertices or edges $e_1$, $e_3$, and $e_{11}$ of an extended fork gadget. All these edges have endpoints which are not adjacent by construction. Hence, there is no path $a \leadsto b$ containing $g$ and, since $O'$ is non-transitive, $e$ is not transitive in $D_{O}(H_\varphi)$. 

\smallskip\noindent($\leftarrow$) It remains to prove the second direction of the proof ($\leftarrow$). Namely, let  $O$ be  a non-transitive $st$-orientation of $H_\varphi$. Let $O'\subset O$ be the restriction to the edges of $G_\varphi$. Since the absence of cycles and of transitive edges are hereditary properties, $D_{O'}(G_\varphi)$ has no cycles and no transitive edges. We have to show that $s$ and $t$ are the only source and sink of $D_{O'}(G_\varphi)$, respectively. Let $v\in G_{\varphi}\setminus \{s,t\}$. We have that if $v$ does not correspond to the vertex denoted $f$ of a fork gadget, then  its neighbourhood in $G_\varphi$ and $H_\varphi$ coincide and, since $D_{O}(H_\varphi)$ is an $st$-graph, then $v$ is  a source or a sink in neither $H_\varphi$ nor $G_\varphi$. Otherwise, by \cref{le:extendedfork} we have that $v$ is incident to at least an edge $e\in G_\varphi$ that enters $v$ and to at least an edge $e'\in G_\varphi$ that exits $v$. Then $v$ is again neither a source nor a sink of $G_\varphi$.
\end{proof}

\subsection{Hardness Subdivisions of Triconnected Graphs with Bounded Degree}
We prove now that \ntsto is \NP-hard even if $G$ is a \emph{4-graph}, i.e., the degree of each vertex is at most 4, and, in addition, it is a subdivision of a triconnected graph.

\smallskip \noindent \textbf{Construction of $J_\varphi$.} Given an instance $\langle X,\varphi\rangle $ of \nae and the instance  $\langle G_\varphi, s,t \rangle$ of \ntsto computed as described in \cref{sse:glimpse}, we compute $\langle J_\varphi, s,t \rangle$ as follows.  
We remove $s$ and $t$ from $J''_\varphi=G_\varphi$. We obtain  a disconnected graph whose  connected components are  $J_{\varphi,1},...,J_{\varphi,h}$. We add a vertex $s_i$ and a vertex $t_i$ to each $J_{\varphi,i}$ (which will play the role of local sources and sinks for each component). Next, for each $i\in [1,h-1]$: (i) We add the edge $(s_i,s_{i+1})$ and $(t_{i+1},t_i)$; 
(ii) We add an edge $e_{{i+1},i}$ incident to a vertex identified as the $f$-vertex of a fork gadget of $J_{\varphi,{i+1}}$ and to a vertex identified as the $f$-vertex of a fork gadget of $J_{\varphi,j}$. 
We denote by $J'_\varphi$ the obtained graph; see \cref{fi:nphardness-deg-a} for a schematic illustration. 
\begin{figure}[t]
    \subfloat[$J'_\varphi$]{\includegraphics[width=0.32\columnwidth,page=8]{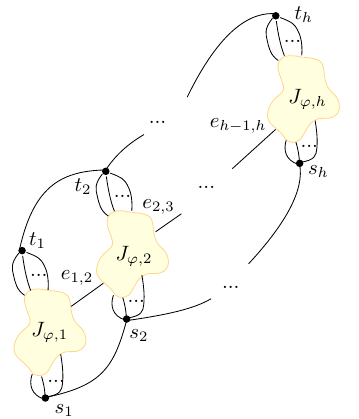}\label{fi:nphardness-deg-a}}
	\hfil
	\centering
	\subfloat[$t_i$ in $J'_{\varphi}$]{\includegraphics[width=0.32\columnwidth,page=6]{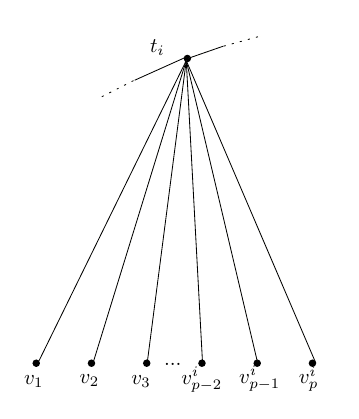}\label{fi:nphardness-deg-b}}
	\hfil
	\subfloat[$t_i$ in $J_{\varphi}$]{\includegraphics[width=0.32\columnwidth,page=7]{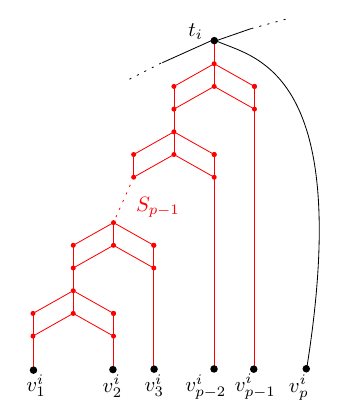}\label{fi:nphardness-deg-c}}
	\caption{(a) Schematic representation of graph $J'_\varphi$. (b-c) How vertex $t_i$ is connected to its neighbourhood in (b)~$J'_\varphi$ and (c)~$J_\varphi$.}\label{fi:nphardness-deg}
\end{figure} 
For each $J_{\varphi,i}$ ($i\in [1,h]$) of  $J'_\varphi$, the only vertices having degree higher than 4 are $s_i$ and $t_i$. For each $i\in [1,h]$, we proceed as follow. We first consider $t_i$. If $t_i$ has degree $p\le 4$, we do nothing. Otherwise, if $p\ge 5$, we proceed as follows: (i) We consider $p-1$ edges incident to $t_i$ and to vertices of $J_{\varphi,i}$ and we remove them; (ii) We connect the endpoints $v_1^i,...,v_p^i$ of the removed edges that are not $t_i$ to a split gadget $S_{p-1}$ and $t_i$ to the input edge of $S_{p-1}$.
\cref{fi:nphardness-deg-b} depicts $t_i$ ($i\in [1,h]$) and its neighborhood $\{v_1^i,...,v_p^i\}$ in $J'_{\varphi}$ and \cref{fi:nphardness-deg-c} depicts how $t_i$ is connected to the vertices $v_1^i,...,v_p^i$ after the above operation.  We perform a symmetric operation on~$s_i$. The resulting graph is denoted by $J_{\varphi}$, and it has vertex degree at most four by construction.

In order to prove \cref{th:degtri}, we first prove the next lemma.

\begin{lemma}\label{le:Jphi3conn}
Graph $J_\phi$ is a subdivision of a triconnected graph.
\end{lemma}
\begin{proof}
We use the same notation as in \cref{fi:nphardness-deg-a}.
Let $\hat{J}_\varphi$ be the graph obtained from $J_\varphi$ by replacing any  two edges $xy$ and $yz$ such that $y$ is a degree-2 vertex with edge $xz$ and removing $y$ from the vertex set.  Let $u,v$ be two vertices of $\hat{J}_\varphi$. We show that there always exist three edge disjoint paths $\pi_1$, $\pi_2$, and $\pi_3$ connecting $u\in J_{\varphi,i}$ to $v\in J_{\varphi,j}$ for $i,j\in [1,h]$. Suppose $i=j$. We show that $J_{\varphi,i}$ ($i\in [1,h]$) is triconnected. Suppose, by contradiction, that $J_{\varphi,i}$ has a separation pair $(a,b)$. Observe that no two vertices of any fork gadget can form a separation pair. Also note that any vertex of a variable gadget is connected to  $s_i$ and $t_i$ by paths that do not include any clause-vertex  (i.e. a vertex representing a clause of the \naelong) or any degree 3-vertex of a variable gadget. It follows that neither $a$ or $b$ can be vertices of variable gadgets nor they can be clause-vertices. Finally, since the pair $(s_i,t_i)$ is not a separation pair, $a$ and $b$ cannot be the source and sink of $J_{\varphi,i}$. It follows that $J_{\varphi,i}$ is triconnected. Suppose $i<j$. We define the following disjoint paths: $\pi_1^1$ and $\pi_2^1$ connect $u$ to $s_i$ and $u$ to $t_i$, respectively; $\pi_1^2=\langle s_is_{i+1},...,s_{j-1}s_j \rangle$ and $\pi_2^2=\langle t_it_{i+1},...,t_{j-1},t_j \rangle$; $\pi_1^3$ and $\pi_2^3$ connect $v$ to $s_j$ and $t_j$, respectively. We set $\pi_1=\pi_1^1\cup \pi_1^2\cup \pi_1^3$ and $\pi_2=\pi_2^1\cup \pi_2^2\cup \pi_2^3$ (with a slight abuse of notation). Observe that $\hat{J}_\varphi\setminus \pi_1\cup \pi_2$ is connected, since, as observed above, each $J_{\varphi,q}$ is triconnected and since each $e_{q,q+1}\not \in \pi_1\cup \pi_2$ ($q\in [1,h]$). Hence, $\pi_3$ can be any path connecting $u$ to $v$ in $\hat{J}_\varphi\setminus \pi_1\cup \pi_2$.
\end{proof}

\begin{theorem}
\label{th:degtri}
\ntsto is \NP-hard for 4-graphs that are subdivisions of triconnected graphs.
\end{theorem}
\begin{proof}
Graph $J_\varphi$ is a $4$-graph by construction and it is triconnected by \cref{le:Jphi3conn}. 
We show that a non-transitive $st$-orientation of $G_\varphi$ can be turned into a non-transitive $st$-orientation of $J_\varphi$ ($\rightarrow$) and vice versa ($\leftarrow$).

\smallskip \noindent ($\rightarrow$)  Given a non-transitive $st$-orientation $O'$ of $G_\varphi$, we compute an orientation $O$ of $J_\varphi$ by extending $O'$ as follow. For any $i\in [1,h]$:
\begin{itemize}
\item We orient the input edge incident to $t_i$ and $s_i$ of the split gadget that we added in the construction of $J_\varphi$ so that it enters $t_i$ and exits $s_i$, respectively. The orientation of all the other edges is given by the properties of the split gadget.
\item We orient $t_it_{i+1}$ from $t_i$ to $t_{i+1}$ and  $s_is_{i+1}$ from $s_i$ to $s_{i+1}$. Also, we orient  $e_{i,i+1}$ from its endpoint in $J_{\varphi,i}$ to its endpoint in $J_{\varphi,i+1}$.
\end{itemize}
We have that $D_O(J_\varphi)$ has one source $s_1$ and one sink $t_h$. Also, observe that for each orientation of the input edge of the split gadget, the split gadget has no cycle. Since there is no edge directed from a vertex in $J_{\varphi,i}$ to a vertex in $J_{\varphi,j}$ for any $i,j\in [1,h]$ so that $i>j$, and since $D_{O'}(G_\varphi)$ had no cycle, we have that $O$ is an $st$-orientation. It remains to show that $D_O(J_\varphi)$ has no transitive edge. Let $uv$ be an edge of $D_O(J_\varphi)$, where $u\in J_{\varphi,i}$ to $v\in J_{\varphi,j}$ ($i,j\in [1,h]$). If $i=j$, observe that there is a path $u\leadsto v$ in $D_{O'}(G_\varphi)$ if and only if the same holds in $D_O(J_\varphi)$. Hence, only edges of the split gadgets can be transitive in $G$, which is not possible, and thus  $uv$ is not transitive. Otherwise, either $uv=t_it_{i+1}$ or $uv=s_is_{i+1}$ or $uv=e_{i,i+1}$. In the first two cases $uv$ is not transitive because there is no path connecting $t_i$ to $t_{i+1}$ or $s_i$ to $s_{i+1}$ different to edge $uv$. Suppose $uv=e_{i,i+1}$. If there exists a directed path  $u \leadsto v$, then it must pass through $s_is_{i+1}$ or $t_it_{i+1}$. The first case is impossible, because there is no directed path connecting $u$ to $s_i$. The second case is also impossible, because there is no directed path connecting $t_i$ to $v$. Hence, $uv$ is not transitive and $O$ is a non-transitive $st$-orientation of $J_\varphi$.

\smallskip \noindent ($\leftarrow$)
Given a non-transitive $st$-orientation $O$ of $J_\varphi$, we compute an orientation $O'$ of $G_\varphi$ as follows. We direct each edge in $G_\varphi$, which are all the edge with the exception of the ones incident to $s$ and $t$, as in $O$. We direct all the other edges entering $t$ or exiting $s$. By \cref{le:extendedfork} $D_{O'}(G_\varphi)$ has only one source $s$ and only one sink $t$. Let $G_{\varphi,i}=G_\varphi \cap J_{\varphi,i}$. Since $D_{O}(J_\varphi)$ is non-transitive, each edge $e\in G_{\varphi,i}$ ($i\in [1,h]$) is not transitive. Also, since every edge incident to $s$ or $t$ has an end-vertex of degree 2, these edges (which are the only ones not in $J_\varphi$) are not transitive. It follows that  $O'$ is a non-transitive $st$-orientation of $G_\varphi$.
\end{proof}

\section{Open Problems}\label{se:conclusions}
%
Several interesting open problems stem from our research. Among them:

\begin{itemize}
    \item Is there an \FPT-algorithm for the \sto problem parameterized by treewidth running in $2^{o(\omega^2)}\cdot \text{poly}(n)$ time?
    \item Does \sto parameterized by treedepth admit a polynomial kernel?
    \item We have shown that finding non-transitive $st$-orientations is \NP-hard for graphs of vertex degree at most four. On the other hand, the problem is trivial for graphs of vertex degree at most two. What is the complexity of the problem for vertex degree at most three? Similarly, one can observe that the problem is easy for graphs of diameter at most two, while it remains open the complexity for diameter in the range $[3,5]$. 
\end{itemize}

\bibliographystyle{plainurl}
\bibliography{bibliography}
\clearpage

\end{document}